\documentclass[%
article, two column,
superscriptaddress,
preprintnumbers,
nofootinbib,
 amsmath,amssymb,amsthm,
 aps,
 prd,
floatfix,
longbibliography.usenames, dvipsnames,10pt]{revtex4-2}
\usepackage{amsthm}
\usepackage[main=english]{babel}
\usepackage{quantikz}

\usepackage[toc,page]{appendix}
\usepackage{amssymb}
\usepackage{graphicx} 
\usepackage{comment}

\usepackage{acro}
\DeclareAcronym{bch}{
  short = BCH ,
  long  = Baker-Campbell-Hausdorff
}
\DeclareAcronym{qsp}{
  short = QSP ,
  long  = Quantum Singular Processing
}
\DeclareAcronym{qsvt}{
  short = QSVT,
  long  = Quantum Singular Value Transformation
}
\DeclareAcronym{lcu}{
  short = LCU,
  long  = Linear Combination of Unitaries
}
\DeclareAcronym{bliss}{
  short = BLISS,
  long  = Block-Invariant Symmetry Shift
}
\DeclareAcronym{lhst}{
  short = LHST,
  long  = Local Hilbert-Schmidt Test
}

\usepackage{subfig}
\usepackage{comment}
\usepackage{dcolumn}
\usepackage{siunitx}
\usepackage{dsfont}

\newcommand\figref{Fig.~\ref}

\usepackage[utf8]{inputenc}
\usepackage[toc,page]{appendix}

\usepackage{ragged2e}
\usepackage{graphicx}
\usepackage{hyperref}
\hypersetup{colorlinks}
\usepackage{soul}

\usepackage{physics}
\usepackage{bbold}
\usepackage{float}
\usepackage{hyphenat}

\usepackage[justification=justified]{caption}
\usepackage{orcidlink}
\usepackage{amssymb}
\newtheorem{theorem}{Theorem}

\definecolor{azure(colorwheel)}{rgb}{0.0, 0.5, 1.0}
\definecolor{babyblue}{rgb}{0.54, 0.81, 0.94}
\definecolor{cornflowerblue}{rgb}{0.39, 0.58, 0.93}
\definecolor{bordeaux}{RGB}{220,20,60}

\hypersetup{colorlinks=true,citecolor=azure(colorwheel),linkcolor=bordeaux,filecolor=bordeaux,urlcolor=azure(colorwheel),breaklinks=true}

\newcommand{\pushright}[1]{\ifmeasuring@#1\else\omit\hfill$\displaystyle#1$\fi\ignorespaces}

\begin{document}

\title{When is randomization advantageous in quantum simulation?}
\author{Francesco Paganelli\orcidlink{0009-0003-5617-2637}} 
\email{francesco.paganelli@cern.ch}
\affiliation{European Organization for Nuclear Research (CERN), Geneva 1211, Switzerland}
\affiliation{Leiden Institute of Physics (LION), Leiden University, Leiden 2333 CA, The Netherlands}
\author{Michele Grossi\orcidlink{0000-0003-1718-1314}}
\affiliation{European Organization for Nuclear Research (CERN), Geneva 1211, Switzerland}
\author{Andrea Giachero\orcidlink{0000-0003-0493-695X}}
\affiliation{Department of Physics, University of Milano Bicocca, Milan, I-20126, Italy}
\affiliation{INFN - Milano Bicocca, Milan, I-20126, Italy}
\affiliation{Bicocca Quantum Technologies (BiQuTe) Centre, Milan, I-20126, Italy}
\author{Thomas E. O'Brien}
\affiliation{Google Quantum AI, Munich, Germany}
\affiliation{Leiden Institute of Physics (LION), Leiden University, Leiden 2333 CA, The Netherlands}
\author{Oriel~Kiss\orcidlink{0000-0001-7461-3342}}
\email{oriel.kiss@meetiqm.com}
\affiliation{Google Quantum AI, Munich, Germany}
\affiliation{European Organization for Nuclear Research (CERN), Geneva 1211, Switzerland}
\affiliation{Department of Nuclear and Particle Physics, University of Geneva, Geneva 1211, Switzerland}
\affiliation{IQM Quantum Computers, Georg-Brauchle-Ring 23-25, 80992 Munich, Germany}

\noaffiliation
\date{\today}

\begin{abstract}

We study the regimes in which Hamiltonian simulation benefits from randomization. We introduce a sparse-QSVT construction based on composite stochastic decompositions, where dominant terms are treated deterministically and smaller contributions are sampled stochastically. Crucially, we analyze how stochastic and approximation errors propagate through block-encoding and QSVT procedures.
To benchmark this approach, we construct ensembles of random Hamiltonians with controlled coefficient dispersion, locality, and number of terms --- designed to favor randomization, and therefore providing an upper bound on its practical advantage.
For Hamiltonians with many terms and highly inhomogeneous coefficient distributions, randomized methods reduce gate counts by up to an order of magnitude. However, this advantage is confined to moderate-precision regimes: as the target error decreases, deterministic methods become more efficient, with a crossover near $\varepsilon \sim 10^{-3}$. Although this regime partially overlaps with quantum chemistry Hamiltonians, realistic systems exhibit additional structure --- such as commutation patterns --- not captured by our model, which are expected to further favor deterministic approaches.

\end{abstract}

\maketitle

\section{Introduction}

The simulation of quantum systems stands among the most promising applications of quantum 
computing~\cite{feynman}. While certain models admit efficient classical treatment, generic 
interacting quantum systems remain computationally intractable. Many-body systems with strong 
correlations or fermionic statistics are particularly challenging: they suffer from exponential 
Hilbert-space growth or, in the case of Monte Carlo approaches, from the sign 
problem~\cite{sign_Troyer}.

Quantum algorithms offer a natural alternative by implementing Hamiltonian time evolution 
directly on quantum hardware. For broad classes of systems --- including local 
Hamiltonians~\cite{Lloyd_1996} and structured many-body models~\cite{ChildsToward} --- the 
required resources scale polynomially in system size and evolution time. This has motivated 
applications across nuclear physics~\cite{Dumitrescu_2018,roggero2019,PRC_Kiss,gqdy-dvps}, 
condensed-matter physics~\cite{andersen2024thermalization,Hofstetter_2018,quantum_simulations_XXZ,miessen2024benchmarking,Kiss2025}, quantum field 
theory~\cite{QS_QFT_Preskill,QS_Schwinger_Lougovski,Klco_2022}, and quantum 
chemistry~\cite{Su2021nearlytight,1stq_Babbush}, 
with representative tasks including energy estimation via quantum phase 
estimation~\cite{QPE-Lloyd} and some near-term variation \cite{LT,PRXQ_Blunt23,EFTQC_practice,QETU}, reaction-rate calculations~\cite{reaction}, correlation 
functions~\cite{two_point_roggero,Hofstetter_2018,quantum_simulations_XXZ,moment_kiss}, 
neutrino dynamics~\cite{PRD_neutrino,neutrino_simulation_amitrano,kiss2025neutrino}, and 
nuclear scattering~\cite{neutrio_nucleus_roggero,du2021,Illa_2022}.

Despite this progress, selecting an appropriate simulation strategy for a given Hamiltonian 
remains nontrivial. Asymptotic error bounds provide useful guidance, but typically depend on 
coarse spectral parameters and can fail to reflect practical performance in physically 
relevant regimes. This gap is especially pronounced for Hamiltonians composed of many Pauli 
terms with strongly inhomogeneous coefficients --- a structure common in quantum chemistry 
(see App.~\ref{app:StatisticalHam}). In such cases, it is not clear \emph{a priori} whether 
deterministic or randomized simulation methods offer a practical advantage, particularly once 
additional structural features such as commutation patterns are taken into account.


The central goal of this work is to identify which structural properties of a Hamiltonian render randomized simulation methods advantageous, and to delineate the regimes in which such advantages are practically relevant. To this end, we study ensembles of random Hamiltonians in which the number of terms, coefficient variance, and locality are varied in a controlled manner. These ensembles are designed to isolate the statistical features most favorable to randomization, and therefore provide an upper bound on its potential benefit.

We benchmark deterministic product-formula methods against randomized alternatives --- namely qDRIFT~\cite{QDrift} and SparSto~\cite{sparsto} --- as well as \ac{qsvt}-based algorithms~\cite{QSP_Chuang}. Alongside this comparison, we develop a general framework for analyzing how errors propagate through the block-encoding and polynomial transformation steps of \ac{qsvt}, quantifying the impact of approximate or stochastic oracle implementations on overall simulation fidelity.

Building on this analysis, we introduce a sparse \ac{qsvt} variant, in which the Hamiltonian block-encoding is replaced by a composite stochastic decomposition. Dominant terms are treated deterministically, while the remaining contributions are sampled, reducing the expected block-encoding cost at the expense of stochastic errors that accumulate through the polynomial sequence.

Taken together, our results provide a systematic comparison of deterministic and randomized simulation methods across a structured family of random Hamiltonians, and demonstrate that any practical advantage of randomization is confined to a well-defined, moderate-precision regime.
The paper is organized as follows. Section~\ref{sec:toolbox} reviews product formulas, randomized techniques, and \ac{qsvt}. Random QSVT methods (previous and our work) are covered in Section~\ref{sec:randQSVT}. Section~\ref{sec:Model} describes the Hamiltonian ensembles, while Section~\ref{sec:results1} presents the empirical error-propagation study and Section~\ref{sec:results2} the comparative of deterministic and randomized methods.

\section{Background on Quantum Simulation Methods}
\label{sec:toolbox}

We consider the problem of simulating real-time evolution generated by a time-independent Hamiltonian
\[
H = \sum_{j=1}^{L} c_j H_j,
\qquad
\mathcal{E}(t) = e^{-iHt},
\]
where each $H_j$ is Hermitian and normalized such that $\|H_j\|\le 1$. We denote by $\lambda = \sum_j |c_j|$ the one norm of the Hamiltonian. The integer $L$ denotes the number of terms in the chosen decomposition, and $t\in\mathbb{R}$ the total evolution time. The decomposition is assumed to reflect the implementation structure of the target architecture, so that exponentials of individual terms $H_j$ can be implemented efficiently.

The simulation task is to approximate the exact time evolution operator $\mathcal{E}(t)$ to target accuracy $\varepsilon$, while minimizing quantum resources such as circuit depth and total gate count. In this work, we focus on two broad and structurally distinct paradigms.

The first class comprises \emph{product-formula} methods, which approximate $\mathcal{E}(t)$ by composing short-time evolutions under individual Hamiltonian terms. Their performance is governed by commutator structure and time discretization error.

The second class consists of \emph{polynomial} methods based on Quantum Singular Value Transformation (QSVT). These approaches approximate the exponential function directly by a polynomial in (a suitably block-encoded version of) $H$, thereby avoiding time slicing.

\subsection{Product formulas}

Product-formula methods constitute one of the earliest and most extensively studied approaches to Hamiltonian simulation. They were first proposed in the context of quantum computation by \textcite{Lloyd_1996}, who suggested exploiting the Lie--Trotter formula~\cite{trotterOriginal} to approximate the time-evolution operator by a product of exponentials of simpler Hamiltonian components. This procedure, commonly referred to as \emph{Trotterization}, was later generalized by Suzuki to arbitrarily high order \cite{Suzuki1976, SuzukiHigherOrderGene}.

We assume a decomposition of the Hamiltonian into efficiently implementable terms, where the operators $H_j$ are typically Pauli strings. The exact time-evolution operator
\(
e^{-iHt}
\)
is approximated by $r$ repetitions of a basic product formula. The first-order (Lie--Trotter) step is
\begin{equation}
\label{eq:Trotter1}
S_1(t/r)
= \prod_{j=1}^{L} e^{-i c_j H_j t / r},
\end{equation}
while the second-order is given by

\begin{equation}
S_{2}(t/r) = \prod_{j=1}^{L} e^{-i c_j H_j t/2r}
   \prod_{j=L}^{1} e^{-i c_j H_j t/2r}, 
\end{equation}
The total time evolution operator is then $S_p(t/r)^r$, with error $\mathcal{O}(\frac{t^{p+1}}{r^{p}})$.
 A systematic analysis of Trotter errors based on commutators was developed in~\textcite{ChildsTrotterError}, refining earlier bounds based on truncations of the Baker--Campbell--Hausdorff expansion~\cite{Raeisi_2012, TrotterOvershootQC}. 
We notice that such bounds are often loose in practice~\cite{EmpricalTrotterFH, ChildsToward}, as they describe worst-case errors.

The circuit depth scales both with the number of steps $r$ and with the number of Hamiltonian terms $L$ applied per step. This motivates stochastic constructions that reduce the number of exponentials implemented at each step. Prominent examples include qDRIFT~\cite{qDRIFT_first} and SparSto~\cite{sparsto}, which replace deterministic iteration over all terms by randomized sampling. We discuss these methods below.

\subsection{Randomization in Product Formulas}
\label{sec:Randomization}

Randomization is a useful tool to reduce circuit depth, by trading with some variance. For instance, randomizing the ordering at each Trotter step reduces the required number of steps for a fixed target accuracy~\cite{ChildsRandomTrotter}. In the following, we consider methods that construct an unbiased estimator of the time evolution operator, and relies on concentration effects to control the simulation error. This approach is believed to be particularly advantageous when $L$ is large or when the coefficients $\{|c_j|\}$ exhibit strong inhomogeneity. 

\subsubsection{qDRIFT}
The qDRIFT protocol~\cite{qDRIFT_first}, subsequently refined in
Refs.~\cite{QDRift_caltech, qSWIFT, qDRIFT_Kiss, david2025tighter_qdrift},
approximates time evolution by sampling Hamiltonian terms according to
their relative weights. We define the sampling probabilities as
\begin{equation}
p_j = \frac{|c_j|}{\lambda}.
\end{equation}
At each step, an index $j$ is sampled independently according to the
distribution $\{p_j\}$, and the corresponding unitary
\begin{equation}
\exp\!\left(
-i\, \frac{c_j}{p_j} H_j \frac{t}{N}
\right)
\end{equation}
is applied. After $N$ independent samples, the resulting approximation
to the time-evolution operator is
\begin{equation}
e^{-iHt}
\approx
\prod_{i=1}^{N}
\exp\!\left(
-i\, \frac{c_{j_i}}{p_{j_i}} H_{j_i} \frac{t}{N}
\right).
\end{equation}
Each step involves only a single Hamiltonian term, so the gate count per step is (explicitly) independent of $L$. The spectral-norm error satisfies~\cite{qDRIFT_first, qDRIFT_errors}
\begin{equation}
\label{eq:qDRIFTerror}
\left\|
e^{-iHt}
-
\prod_{i=1}^{N}
\exp\!\left(
-i\, \frac{c_{j_i}}{p_{j_i}} H_{j_i} \frac{t}{N}
\right)
\right\|
\le
\mathcal{O}\!\left(
\frac{\lambda^2 t^2}{N}
\right).
\end{equation}

\subsubsection{SparSto}

SparSto~\cite{sparsto} interpolates between deterministic Trotterization and qDRIFT by sampling a sparsification $\hat{H}$ of the Hamiltonian at each time step, where a sparsification is composed by only a subset of all the terms. This is similar to the hybrid method from Refs. ~\cite{Hybrid_Wiebe, gunther2025phase,qDRIFT_Kiss,hybrid_jin}, but with an error bound that explicitly depend on the probability distribution.  For a probability vector $\vec{p} = (p_1,\ldots,p_L)$ with $0 < p_j \le 1$, define a random Hamiltonian estimator
\[
\hat{H}
=
\sum_{j=1}^{L}
\frac{c_j}{p_j} H_j \, \xi_j,
\]
where $\xi_j \sim \mathrm{Bernoulli}(p_j)$ are independent random variables. Then $\mathbb{E}[\hat{H}] = H$. The expected number of terms per step is
\[
\mu = \sum_{j=1}^{L} p_j,
\]
which interpolates between $\mu=L$ (deterministic product formula) and $\mu=1$ (qDRIFT).

In each time slice of duration $\delta t$, a fresh realization of $\hat{H}$ is sampled and exponentiated via a first-order product formula. Repeating this procedure yields a stochastic approximation to $e^{-iHt}$.

A key design choice is the probability vector $\vec{p}$. Following~\cite{sparsto}, we partition the terms into a set $A$ of dominant contributions and its complement $\bar{A}$. For $j \in A$, we set $p_j = 1$, while for $j \in \bar{A}$ we set $p_j = \alpha |c_j|$, with $\alpha$ chosen so that $p_j \le 1$ for all $j$. This choice exploits the typically heavy-tailed distribution of coefficients in molecular Hamiltonians.

Denoting by $\mathcal{S}(t)$ the resulting quantum channel and by $e^{t\mathcal{L}}$ the exact Liouvillian evolution, the diamond-norm error
\[
\varepsilon_{\mathrm{SparSto}}
=
\|\mathcal{S}(t) - e^{t\mathcal{L}}\|_\diamond
\]
can be upper bounded by ~\cite{sparsto}
\begin{equation}
\label{eq:SparStoError}
\varepsilon_{\mathrm{SparSto}}
=
\frac{2 t^2 \mu}{G} \|\vec{u}\|_1
+
\frac{4 t^{3} \mu^{2}}{3 G^{2}} K
+
\mathcal{O}\!\left(
\frac{t^{4} \mu^{3}}{G^{3}}
\right),
\end{equation}
where
\begin{equation}
\label{eq:u}
\vec{u}
=
\bigl(
(\tfrac{1}{p_1}-1)c_1^2,\ldots,
(\tfrac{1}{p_L}-1)c_L^2
\bigr),
\end{equation}
$G$ is the total number of gates and $K$ collects higher-order contributions depending on $\vec{p}$ and $\vec{c}$ as defined in~\cite{sparsto}.

The bound is minimized for $p_j=1$ for all $j$, corresponding to deterministic Trotterization. However, this choice maximizes $\mu=L$, and therefore the per-step circuit depth. By tuning $\vec{p}$, SparSto enables explicit control over $\mu$, allowing a trade-off between per-step cost and stochastic error. In regimes with heavy-tailed coefficient distributions, this trade-off can reduce overall resource requirements compared to both deterministic product formulas and fully stochastic qDRIFT. 

While not considered explicitly in this work, several other approaches that leverage randomization within product-formula constructions have been proposed, including Te-PAI \cite{Granet2024,Kiumi2025} and randomized multi-product formulas \cite{Faehrmann2022randomizingmulti}.
\subsection{QSVT-based simulation}

Quantum Singular Value Transformation (QSVT) provides a fundamentally different approach to Hamiltonian simulation based on polynomial transformations of block-encoded operators. It generalizes the Quantum Signal Processing (QSP) framework of~\textcite{LowQubitization, LowQSP}, originally formulated for single-qubit unitaries, to arbitrary matrices via their singular values~\cite{GilBeyond, unification}. In this framework, a polynomial satisfying specific magnitude and parity constraints is applied to the singular values of a matrix encoded within a larger unitary.

A QSVT circuit alternates between two types of operations. The first is a block-encoding $U_A$, a unitary whose action embeds a matrix $A$ into a subspace identified by projectors $\Pi$ and $\tilde\Pi$. The second consists of projector-controlled phase gates $\Pi_{\phi_j}$ and $\tilde\Pi_{\phi_j}$, where the phase angles $\vec{\phi} = \{\phi_j\}$ are computed classically from the coefficients of the target polynomial. Each such operation applies a phase $e^{i\phi_j}$ to the block-encoding subspace and $e^{-i\phi_j}$ to its orthogonal complement. These operations can be written as~\cite{GilBeyond}
\begin{equation}\label{eq:ProjectorControlledPhases}
    \Pi_\phi = e^{i\phi(2\Pi - \mathbb{I})},
    \qquad
    \tilde{\Pi}_\phi = e^{i\phi(2\tilde{\Pi} - \mathbb{I})}.
\end{equation}



Hamiltonian simulation via QSVT proceeds by approximating the function
\[
e^{-ixt} = \cos(xt) - i \sin(xt)
\]
with a polynomial that can be implemented using the above construction. Following~\cite{GilBeyond}, the cosine and sine functions are expanded using Jacobi--Anger expansions,
\begin{align}\label{eq:JAExpansion}
\begin{split}
    \cos(xt) &\approx J_0(t) + 2 \sum_{k=1}^{d} (-1)^k J_{2k}(t) T_{2k}(x), \\
    \sin(xt) &\approx 2 \sum_{k=0}^{d} (-1)^k J_{2k+1}(t) T_{2k+1}(x),
\end{split}
\end{align}
where $J_k$ denotes the $k$th Bessel function of the first kind and $T_k$ is the $k$th Chebyshev polynomial. The expansions are truncated at degree $d$ to achieve a target simulation error $\varepsilon$, and the resulting coefficients are used to compute the phase sequence $\vec{\phi}$.

To apply QSVT to a Hamiltonian $H$, one must first construct a block-encoding. Assuming that $H$ admits a linear combination of unitaries (LCU) decomposition,
\[
H = \sum_{j=0}^{L-1} c_j H_j,
\]
a block-encoding can be implemented using the standard LCU primitive~\cite{LCU_childs}. This construction requires $a = \lceil \log_2 L \rceil$ ancilla qubits, a \textsc{Prepare} oracle $V$, and a \textsc{Select} oracle $S$. The \textsc{Prepare} unitary acts as
\begin{equation}
    V \ket{0}^{\otimes a}
    = \frac{1}{\sqrt{\lambda}} \sum_{j=0}^{L-1} \sqrt{|c_j|} \ket{j},
\end{equation}
where $\lambda = \|H\|_1 = \sum_j |c_j|$. The \textsc{Select} oracle applies the corresponding Hamiltonian term conditioned on the ancilla state,
\begin{equation*}
    S = \sum_{j=0}^{L-1} \ket{j}\bra{j} \otimes H_j.
\end{equation*}
Combining these operations yields the block-encoding
\begin{equation*}
    W = (V^\dagger \otimes I_n)\, S \, (V \otimes I_n)
    =
    \begin{pmatrix}
        H / \lambda & * \\
        * & *
    \end{pmatrix},
\end{equation*}
where $I_n$ denotes the identity on the $n$ system qubits.

Due to this normalization, simulating the evolution under $H$ for time $t$ requires implementing QSVT for an effective evolution time $\lambda t$. The truncation degree $d$ required to approximate $e^{-ix\lambda t}$ by a polynomial with error $\varepsilon_\text{poly}$ is defined implicitly by~\cite{unification}
\begin{equation}\label{eq:QSVT_poly_error}
    \varepsilon_\text{poly} = \left(\frac{\lambda t}{d(\varepsilon_\text{poly}, \lambda t)}\right)^{d(\varepsilon_\text{poly}, \lambda t)}.
\end{equation}
As $\varepsilon_\text{poly}$ is the only source of error, the total simulation error satisfies $\varepsilon \approx \varepsilon_\text{poly}$, and the required truncation degree scales asymptotically as~\cite{LowQSP, unification}
\begin{equation}\
\label{eq:JATruncationDegree}
    d =
    \Theta\!\left(
        \lambda t +
        \frac{\log(1/\varepsilon)}{
        \log\!\left(e + \frac{\log(1/\varepsilon)}{\lambda t}\right)}
    \right).
\end{equation}

QSVT exhibits asymptotically superior error scaling, depending only logarithmically on $1/\varepsilon$, whereas product-formula methods scale polynomially as $\varepsilon^{-1/p}$. In practice, QSVT typically incurs larger constant overheads due to the depth of the LCU circuit and the additional ancilla qubits required.

\section{Randomized QSVT}
\label{sec:randQSVT}
\subsection{Previous work}

Randomized variants of quantum singular value transformation and related techniques have already been explored in the literature. For example, Ref.~\cite{RandomQSVT} proposes a randomized version of QSVT that avoids constructing a full block-encoding of the target operator. Instead of assuming coherent access to a block-encoding of $H$, the method assumes sampling access to terms $H_k$ from an LCU decomposition
\[
H = \sum_k \lambda_k H_k .
\]
The authors introduce the two--by--two operator
\begin{equation}
U =
\begin{pmatrix}
cI & sH \\
sH^\dagger & -cI
\end{pmatrix},
\qquad c^2 + s^2 = 1,
\end{equation}
which acts as a randomized analogue of a block encoding. By sampling $H_k$ according to the distribution $|\lambda_k|$, one can implement unitaries
\begin{equation}
U_k =
\begin{pmatrix}
cI & sH_k \\
sH_k^\dagger & -cI
\end{pmatrix},
\end{equation}
that satisfy $\mathbb{E}[U_k] = U$. Replacing each occurrence of $U$ in the alternating phase sequence of QSVT with an independent sample $U_k$ yields a randomized circuit whose expectation reproduces the desired polynomial transformation of $H$.

A related approach is developed in Ref.~\cite{stat_QPE}, where the time-evolution operator is expressed as a linear combination of unitaries. Assuming a Hamiltonian written as a convex combination of Pauli operators, the construction begins by decomposing the evolution operator as $e^{iHt}=(e^{iHt/r})^r$. Each short-time propagator is then Taylor expanded and regrouped into linear combinations of Pauli rotations, resulting in an LCU representation that can be sampled to estimate expectation values involving time evolution.

Finally, Ref.~\cite{SCU} proposes a stochastic Hamiltonian simulation algorithm based on a convex decomposition of the Taylor expansion of the time-evolution operator. In this approach, the truncated expansion can be rewritten as
\begin{equation}
e^{-iHt} = L_c \sum_j p_j P'_j + p_1 + L_s^2 \sum_k p_k e^{i\theta P'_k},
\end{equation}
where \(P'_j\) are Pauli strings (up to a sign) and the coefficients define a probability distribution after normalization. This formulation enables \emph{convex Taylor sampling}, where the time evolution is approximated by repeatedly sampling unitary terms from the resulting convex combination.

\subsection{Sparse QSVT}
\label{sec:sparse_qsvt}

In the following, we introduce a sparse-QSVT construction in which the Hamiltonian $H$ appearing in the LCU block-encoding is replaced by a stochastic estimator $\hat{H}$. Conceptually, this approach is related to randomized QSVT frameworks~\cite{RandomQSVT}. However, rather than randomizing the block-encoding directly, we construct $\hat{H}$ using the SparSto technique introduced in Ref.~\cite{sparsto}. Our primary focus is not the construction itself, but rather the analysis of how stochastic and approximation errors propagate through the LCU and QSVT procedures.

In the following, all error bounds should be interpreted as worst-case upper bounds on the expected error, assuming independent sampling in each use of the stochastic estimator. Correlations between different applications of the block-encoding are neglected, and may lead to improved performance in practice.

\begin{theorem}[LCU error propagation]
\label{theo:MainLCUPropagation}
Let $V$ be the exact \textsc{Prepare} oracle and $\tilde V$ an approximation satisfying
\begin{equation}
\label{eq:prep_error}
\varepsilon_{\mathrm{prep}} = \|V-\tilde V\|.
\end{equation}
Then the induced error on the LCU block-encoding
\begin{equation}
\tilde W = (\tilde V^\dagger \otimes I)\, S\, (\tilde V \otimes I)
\end{equation}
satisfies
\begin{equation}
\label{eq:lcu_error_bound}
\varepsilon_{\mathrm{LCU}}
=
\|W-\tilde W\|
\le
2\,\varepsilon_{\mathrm{prep}}.
\end{equation}
An analogous linear bound holds for \textsc{Select}.
\end{theorem}

The first source of error arises from approximations in the LCU primitives. Since the block-encoding is constructed by conjugating the \textsc{Select} operator with the \textsc{Prepare} unitary, any imperfection in these components propagates to the resulting block-encoding. The above theorem shows that this propagation is stable: the induced error grows at most linearly with the preparation error.

\begin{theorem}[QSVT error amplification]
\label{theo:MainQSVTPropagation}
Let $U_A$ be an exact block-encoding and $\tilde U_A$ an approximation satisfying
\begin{equation}
\label{eq:be_error}
\varepsilon_{\mathrm{BE}} = \|U_A-\tilde U_A\|.
\end{equation}
Then a degree-$d$ \ac{qsvt} sequence obeys
\begin{equation}
\label{eq:qsvt_error_bound}
\varepsilon_{\mathrm{QSVT}}
\le
d\,\varepsilon_{\mathrm{BE}}
+
\varepsilon_{\mathrm{poly}},
\end{equation}
where $\varepsilon_{\mathrm{poly}}$ denotes the polynomial truncation error.
\end{theorem}

Once a block-encoding is available, QSVT applies a sequence of alternating reflections and phase rotations. Each use of the block-encoding introduces an error, which accumulates across the $d$ steps of the polynomial transformation. Consequently, the overall QSVT error grows at most linearly with the polynomial degree.

For stochastic block-encodings, the dominant contribution to the block-encoding error is controlled by the dispersion of the coefficients entering the estimator. Specifically:

\begin{theorem}[Estimator variance proxy]
\label{theo:MainEstimatorVariance}
Let
\begin{equation}
\label{eq:stochastic_estimator}
\hat H = \sum_j \frac{c_j}{p_j} H_j \,\xi_j,
\end{equation}
be a stochastic estimator of the Hamiltonian constructed using SparSto, where $\xi_j\sim\mathrm{Bernoulli}(p_j)$ independently. Then the quantity
\begin{align}
\label{eq:variance_result}
\mathrm{Var}_{\mathrm{coeff}}[\hat H]
:=
\|\vec u\|_1,
\qquad
\vec u \text{ as defined in \eqref{eq:u}},
\end{align}
captures the coefficient dispersion of the estimator and serves as a proxy controlling the induced block-encoding error.
\end{theorem}

We emphasize that $\mathrm{Var}_{\mathrm{coeff}}[\hat H]$ is a scalar quantity derived from the coefficients, and should not be interpreted as a full operator variance. Rather, it quantifies the magnitude of stochastic fluctuations entering the LCU construction.

Under this approximation, the induced block-encoding error scales as
\begin{equation}
\label{eq:be_stochastic_scaling}
\varepsilon_{\mathrm{BE}}
=
\mathcal{O}\!\left(
\sqrt{\mathrm{Var}_{\mathrm{coeff}}[\hat H]}
\right).
\end{equation}

Combining the previous bounds yields the overall error scaling of Sparse QSVT,
\begin{equation}
\label{eq:sparse_qsvt_total_error}
\varepsilon_{\mathrm{SparseQSVT}}
=
\mathcal{O}\!\left(
d\,\sqrt{\mathrm{Var}_{\mathrm{coeff}}[\hat H]}
\right)
+
\varepsilon_{\mathrm{poly}}.
\end{equation}

Unlike the original SparSto setting, where stochastic errors can average out across repeated queries, in Sparse QSVT the errors propagate coherently through the polynomial transformation. As a result, they accumulate linearly with the degree $d$. Sparse QSVT therefore trades a reduced per-query cost $\Theta(\mu)$ for a controlled linear amplification of the stochastic block-encoding error. Importantly, this framework is not restricted to stochastic sparsification. It also applies to approximate or variational implementations of the \textsc{Prepare} and \textsc{Select} oracles, where the exact LCU primitives are replaced by parameterized or learned circuits. In such settings, the resulting approximation errors enter the block-encoding in the same manner as stochastic errors, and are subsequently amplified through the QSVT sequence. This provides a systematic way to assess the impact of imperfect oracle implementations, and highlights the sensitivity of QSVT-based methods to such approximations. Detailed proofs are provided in Appendix~\ref{app:proof}.

\section{Models}
\label{sec:Model}

Our objective is to understand for which classes of Hamiltonians randomized simulation methods provide a tangible advantage, and to identify the regime in which such advantages persist in practice. We focus on two key structural properties: the number of terms in the Hamiltonian and the dispersion of the coefficient distribution. We effectively construct a parametric family of random Hamiltonians that allows us to systematically vary these properties while keeping the Hilbert-space dimension fixed.

We define a family of Hamiltonians parameterized by the number of terms $L$, the Pauli weight $k$, and a coefficient distribution $\mathcal{P}(c_l)$:
\begin{equation}
\label{eq:RandomHamiltonianExtraction}
\begin{split}
H(L,k)
&=
\sum_{l=1}^{L} c_l \, P^{(k)}_l,
\qquad
c_l \sim \mathcal{P}(c_l), \\
P^{(k)}_l
&\in \{I,X,Y,Z\}^{\otimes N}
\quad \text{with weight } k, \\
&\text{s.t. } \sum_{l} |c_l| = 1.
\end{split}
\end{equation}

Each Pauli string $P^{(k)}_l$ acts non-trivially on up to $k$ qubits. The coefficients are sampled independently from $\mathcal{P}(c_l)$ and rescaled to satisfy $\|H\|_1=\sum_l |c_l|=1$, thereby fixing the overall energy scale. The maximum number of distinct terms is $\sum_{j=1}^{k}\binom{N}{j}3^j$.

The purpose of this ensemble is not to faithfully reproduce physical Hamiltonians, but rather to isolate statistical features that are favorable to randomized simulation methods. In particular, by sampling coefficients independently and neglecting structured commutation patterns, we suppress features—such as locality-induced commutativity—that are known to benefit deterministic approaches. Consequently, any observed advantage of randomized methods should be interpreted as an upper bound on their practical usefulness.

Although physical Hamiltonians possess additional structure (e.g., locality, symmetries, and nontrivial commutation relations), this construction allows us to probe the effect of coefficient dispersion and term count in isolation.

We consider two classes of coefficient distributions: a lognormal distribution and a Pareto type-II (Lomax) distribution. In each numerical experiment, coefficients are sampled from one of these distributions.

\begin{equation}
\mathrm{LogNormal}(x;\mu,\sigma)
=
\frac{1}{x\sigma\sqrt{2\pi}}
\exp\!\left(
-\frac{(\ln x-\mu)^2}{2\sigma^2}
\right),
\end{equation}

\begin{equation}
\mathrm{Pareto\text{-}II}(x;a)
=
a(1+x)^{-(a+1)}.
\end{equation}

The lognormal distribution produces coefficients with broadly dispersed magnitudes, controlled by the variance parameter $\sigma^2$, while the Pareto-II distribution generates heavy-tailed behavior governed by the shape parameter $a$.

In Appendix~\ref{app:StatisticalHam}, we show that molecular Hamiltonians exhibit coefficient statistics that are partially consistent with these distributions. However, real systems also exhibit additional structure—such as correlated terms and nontrivial commutation patterns—not captured by the present model.

Sampling from these distributions allows us to control the degree of coefficient inhomogeneity. Randomized simulation methods are expected to benefit when the distribution is highly uneven, as a small subset of dominant terms can be treated deterministically while smaller contributions are sampled. By varying $\sigma^2$ or $a$, the ensemble in Eq.~\eqref{eq:RandomHamiltonianExtraction} enables a systematic exploration of when such structural imbalance leads to practical advantages for randomized methods.
\section{Results} 

\subsection{Error propagation}
\label{sec:results1}
\begin{figure}[t]
    \centering
    \includegraphics[width=0.9\columnwidth]{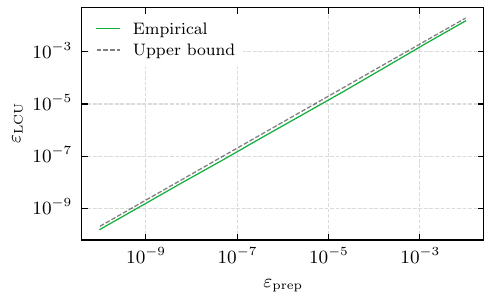}
    \caption{\justifying{\textbf{LCU error under \textsc{Prepare} error.} Block-encoding error as a function of the \textsc{Prepare} oracle error $\varepsilon_{\mathrm{prep}}=\|V-\tilde{V}\|$. Theoretical upper bound and numerical results are shown for comparison. }}
    \label{fig:LCUpropagation}
\end{figure}

\begin{figure}[t]
    \centering
    \includegraphics[]{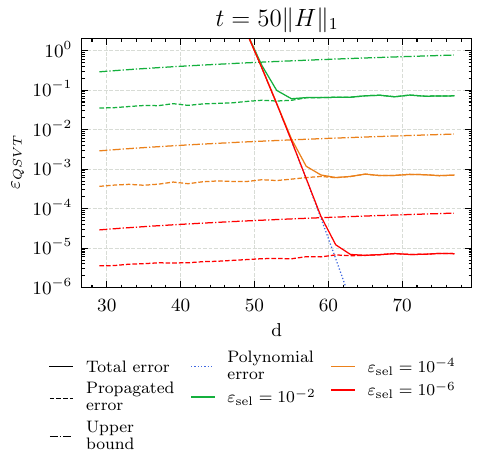}
    \caption{\justifying{\textbf{QSVT error under \textsc{Select} approximation.} Total \ac{qsvt} simulation error under \textsc{Select} approximation with $\varepsilon_{\mathrm{sel}}=\|S-\tilde{S}\|$. The error is decomposed into the analytically computed polynomial truncation contribution~\cite{unification} and the numerically evaluated block-encoding contribution. Results are shown for a random Hamiltonian with Pareto-distributed coefficients ($a=0.9$) at several values of $\varepsilon_{\mathrm{sel}}$.}}
    \label{fig:SelectPropagation}
\end{figure} 

We begin by empirically assessing the tightness of the theoretical error bounds derived in Section \ref{sec:sparse_qsvt}.

\begin{figure*}[tbh]
    \centering
    \includegraphics[]{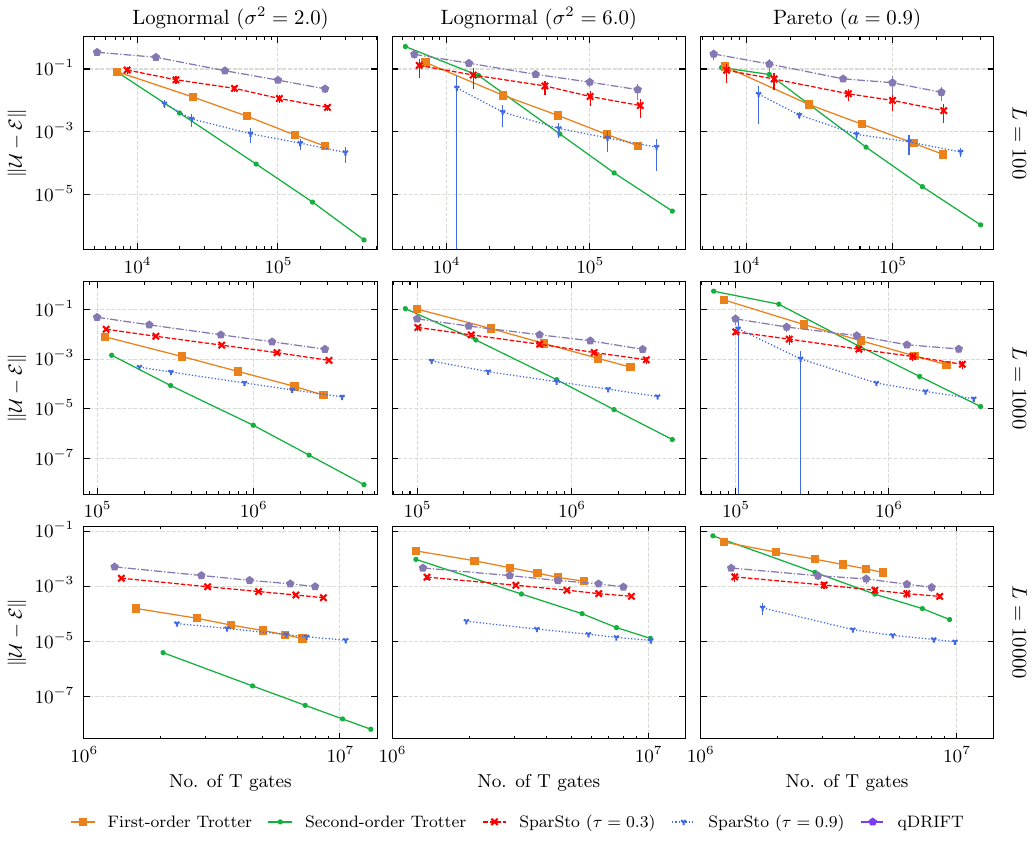}
  \caption{\justifying{\textbf{Trotter vs SparSto across Hamiltonian structure.} 
Spectral error \eqref{eq:spectral_error} as a function number of T gates for first- (orange) and second-(green) order Trotterization, qDRIFT (purple), and SparSto at sparsification thresholds $0.3$ (red) and $0.9$ (blue).
Rows correspond to increasing numbers of terms $L$, and columns to increasing coefficient variance.}}
\label{fig:TrotterSparstroStructure}
\end{figure*}

In \figref{fig:LCUpropagation}, we study the impact of approximating the \textsc{Prepare} oracle on the accuracy of the \ac{lcu} block-encoding of a random $4$-qubit Hamiltonian with coefficients sampled from a Pareto distribution with $a=0.9$. Starting from the exact \textsc{Prepare} oracle $V$, we construct a perturbed operator $\tilde{V}$ satisfying $\varepsilon_{\mathrm{prep}}=\|V-\tilde{V}\|$. Using $\tilde{V}$ together with the exact \textsc{Select} oracle $S$, we explicitly construct the dense matrix representation of the corresponding \ac{lcu} circuit $\tilde{W}$. The block-encoding error $\varepsilon_{\mathrm{LCU}}$ is quantified as the spectral distance between $\tilde{W}$ and the exact block-encoding $W$. Repeating this procedure for increasing values of $\varepsilon_{\mathrm{prep}}$ yields the trend shown in \figref{fig:LCUpropagation}. The observed scaling closely follows the linear bound in~\eqref{eq:lcu_error_bound}.

In \figref{fig:SelectPropagation}, we analyze the error introduced by approximating the \textsc{Select} oracle within a full \ac{qsvt} simulation of the same $4$-qubit Hamiltonian, for evolution time $t=50\tau$. We perturb the exact \textsc{Select} oracle to obtain $\tilde{S}$ such that $\varepsilon_{\mathrm{sel}}=\|S-\tilde{S}\|$, and use it to construct the corresponding \ac{lcu} block-encoding and the associated \ac{qsvt} circuit for different truncation degrees $d$ and values of $\varepsilon_{\mathrm{sel}}$. The total error is separated into two contributions: the polynomial truncation error, evaluated analytically from~\eqref{eq:QSVT_poly_error}, and the error due to the imperfect block-encoding, obtained numerically from the dense matrix representation. The latter exhibits the $\mathcal{O}(d)$ scaling predicted by~\eqref{eq:qsvt_error_bound}, while remaining approximately one order of magnitude below the theoretical bound in the explored parameter regime.

\begin{figure*}[t]
    \centering
    \includegraphics[]{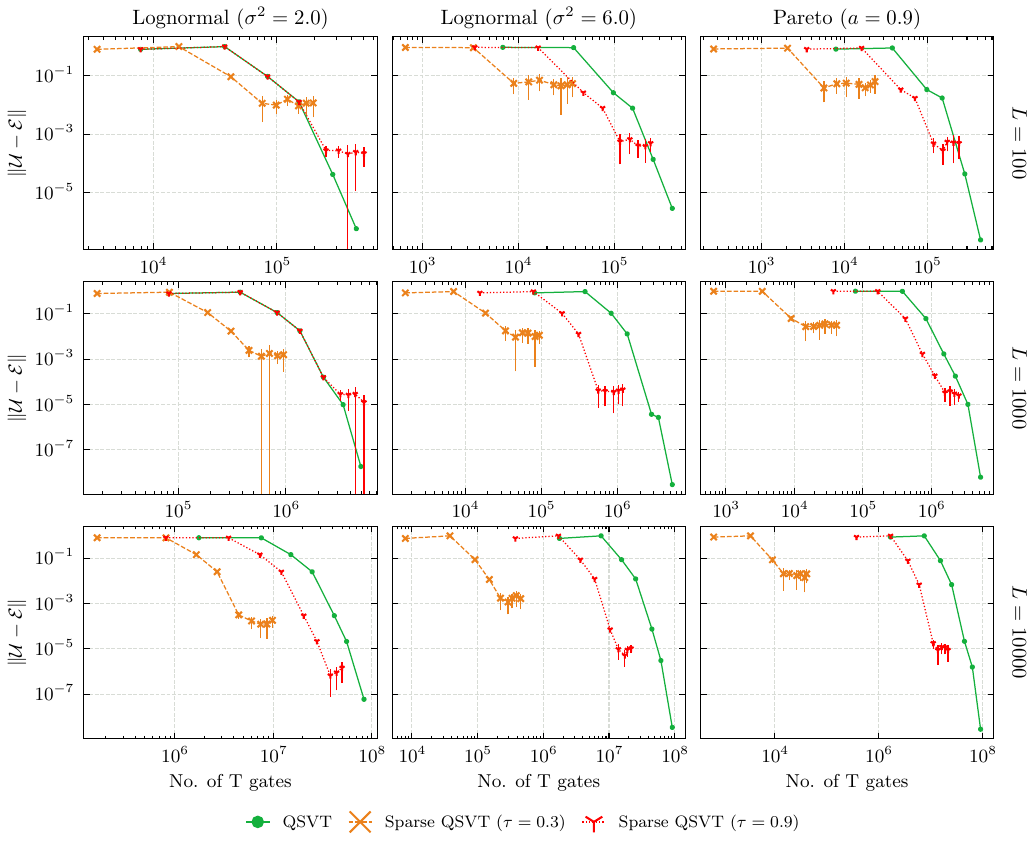}
    \caption{\justifying{\textbf{QSVT vs Sparse QSVT across Hamiltonian structure.} Spectral error \eqref{eq:spectral_error} as a function of number of T gates for \ac{qsvt} (green) and Sparse \ac{qsvt} at sparsification thresholds $0.3$ (orange) and $0.9$ (red). Rows correspond to increasing numbers of terms $L$, and columns to increasing coefficient unevenness. The first two columns use lognormal coefficient distributions, while the last uses a Pareto distribution.}}
    \label{fig:qsvt_hamiltonian_structure}
\end{figure*}

\subsection{Randomized Simulation}
\label{sec:results2}

We now investigate the central question of this work: which Hamiltonians benefit from randomized simulation methods, and in which regimes such benefits are practically relevant. We employ the random Hamiltonians introduced in Sec.~\ref{sec:Model}, which are constructed to emphasize coefficient dispersion and large term counts. As discussed previously, this ensemble suppresses structured commutation patterns and should therefore be interpreted as favorable to randomized methods; any observed advantage should be viewed as an upper bound on their practical usefulness.

We fix the number of qubits to $N=8$ and the Pauli weight to $k=6$. To quantify the simulation error, we use the spectral norm of the difference between the approximate and exact evolution operators. In practice, we compute this as the maximum absolute eigenvalue of the operator difference,
\begin{equation}
\|r\| := \max \{ |\lambda| : \lambda \in \mathrm{Spectrum}(r) \},
\label{eq:spectral_error}
\end{equation}
where $r = \mathcal{U} - \mathcal{E}$. While this is not the usual choice, this metric provides a consistent and efficiently computable proxy for the operator norm in the regimes considered.

Figures~\ref{fig:TrotterSparstroStructure} and~\ref{fig:qsvt_hamiltonian_structure} show results for $8$-qubit Hamiltonians with $L=10^2$ to $10^4$ terms. Coefficients are sampled either from lognormal distributions with $\sigma^2=2.0$ and $\sigma^2=6.0$, or from a Pareto distribution with shape parameter $a=0.9$. For each instance, we perform fixed-time evolution at $t=10\lambda$ and report the spectral error as a function of T gate counts for both deterministic and randomized algorithms. The T gate counts are obtained using the gridsynth algorithm~\cite{RossSelingerAlgorithm}, as described in Appendix~\ref{app:gate_counts}. Error bars represent one standard deviation over 10 Hamiltonian instances.

In Fig.~\ref{fig:TrotterSparstroStructure}, we compare first- and second-order Trotterization with qDRIFT and SparSto for sparsification thresholds $\tau \in \{0.3, 0.9\}$. The parameter $\tau$ specifies the fraction of the total $\ell_1$ norm assigned to the deterministic subset,
\begin{equation} 
\sum_{j \in A} |c_j| \lesssim \tau \lambda.
\end{equation}

For sufficiently low target error, deterministic Trotterization eventually outperforms randomized methods. Increasing either the number of terms or the coefficient heterogeneity shifts the crossover point to lower error, enlarging the regime in which randomization provides an advantage. Heavy-tailed (Pareto) distributions favor randomized schemes more strongly. For the largest system considered, SparSto with threshold $0.9$ remains more efficient than second-order Trotterization up to $\varepsilon \lesssim 10^{-5}$. Among randomized approaches, SparSto consistently outperforms qDRIFT due to smaller constant overheads, while decreasing the threshold brings its performance closer to qDRIFT.

\begin{figure}
\centering 
\includegraphics[width=0.96\linewidth]{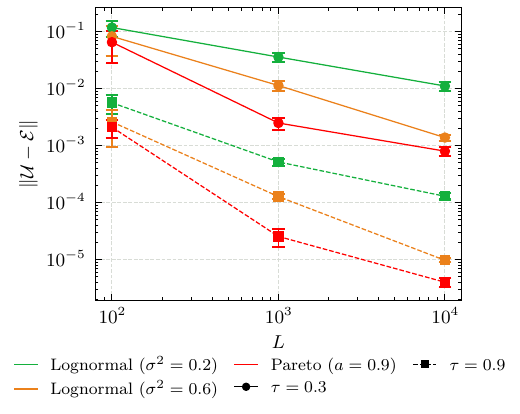} \caption{\justifying{\textbf{Intersection.} This figure shows the point where the error of Sparstro is equal to the second-order Trotter as a function of the number of terms for different variances. This identifies the point at which the deterministic method begins to outperform the stochastic approach.}}
\label{fig:intersection}
\end{figure}

Figure~\ref{fig:qsvt_hamiltonian_structure} presents the analogous comparison between 
\ac{qsvt} and Sparse \ac{qsvt}. Sparse \ac{qsvt} reduces constant overhead in low-accuracy 
regimes without modifying the asymptotic scaling. The stochastic block-encoding error 
accumulates with polynomial degree $d$, introducing an accuracy floor once the 
Jacobi--Anger truncation error becomes negligible. This floor is governed primarily by the 
coefficient distribution and sparsification threshold, rather than directly by $L$. 
Increasing $L$ enables larger gate-count reductions, while increasing coefficient variance 
lowers the attainable accuracy but amplifies the relative reduction in gate cost.

For threshold $0.3$, the minimum achievable error lies between $10^{-1}$ and $10^{-3}$ 
depending on the distribution, with up to three orders of magnitude reduction in $T$-count 
for the largest Pareto instance at errors around $10^{-2}$. Raising the threshold to $0.9$ 
lowers the error floor to $10^{-6}$, but reduces the $T$-count improvement to roughly one 
order of magnitude in the largest Pareto case. When the threshold is $0.9$ and both $L$ and 
the coefficient variance are small (e.g., $L \leq 1000$), no significant gate reduction is 
observed. This behavior stems from the discrete nature of the resource model: gate counts 
scale with the number of ancilla qubits $\lceil\log_2 L\rceil$, so improvements occur only 
when sparsification reduces $L$ sufficiently to cross a power-of-two boundary.

\begin{figure*}[t] \centering 
\includegraphics[width=\linewidth]{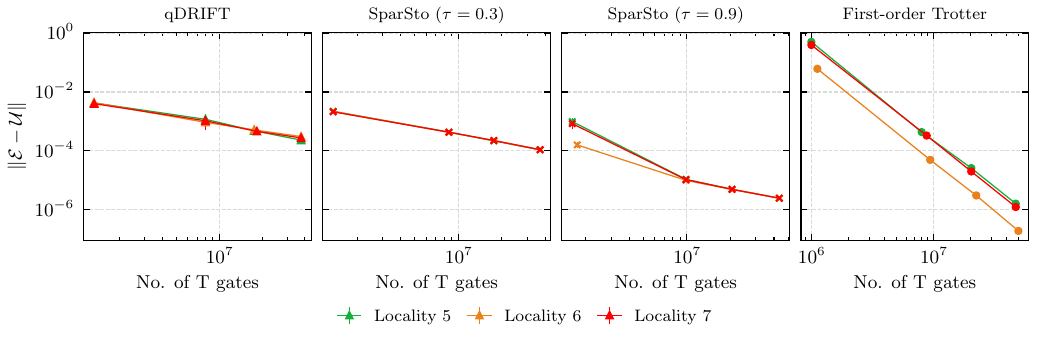} 
\caption{\textbf{Effect of locality.} Simulation error of qDRIFT, Trotter, and SparSto as 
a function of $T$-gate count for varying term locality $k$.} 
\label{fig:LocalityScalingTrotter}
\end{figure*}

Overall, randomization can yield order-of-magnitude reductions in gate count, but only in 
moderate-precision regimes, typically around $\varepsilon \sim 10^{-2}$. At higher 
precision, accumulated stochastic errors offset these gains, rendering randomized approaches 
less competitive.

For product-formula methods (Fig.~\ref{fig:TrotterSparstroStructure}), a clear crossover 
emerges between randomized and deterministic strategies: randomized methods are more 
efficient at low accuracy, while deterministic methods dominate at high accuracy. 
Figure~\ref{fig:intersection} summarizes this behavior by showing the crossover point 
between SparSto and second-order Trotterization as a function of $L$ and coefficient 
variance. No analogous crossover is observed for \ac{qsvt}.

Finally, we examine the effect of the locality parameter $k$ on product-formula simulations, 
considering Hamiltonians drawn from Eq.~\eqref{eq:RandomHamiltonianExtraction} with fixed 
locality $k$ --- that is, with exactly $k$ non-trivial Pauli operators per term.

Figure~\ref{fig:LocalityScalingTrotter} shows that randomized methods are largely 
insensitive to locality within this model. Deterministic methods, by contrast, exhibit a 
mild dependence on the parity of $k$. A plausible explanation lies in the commutator 
structure: Pauli strings commute when the number of anticommuting sites is even, an event 
that occurs with slightly higher probability for even locality. This effect is specific to 
the unstructured ensemble considered here and may not persist in more structured physical 
systems.

\section{Conclusions}

Selecting an appropriate quantum simulation primitive \emph{a priori} remains a nontrivial 
task. Asymptotic complexity alone is often insufficient to predict practical performance, 
as constant factors, coefficient dispersion, and implementation overhead can dominate in 
realistic regimes. In this work, we characterize the conditions under which randomized 
simulation methods offer a practical advantage over deterministic alternatives.

Our analysis shows that randomization is most effective when the number of terms $L$ is 
large and the coefficient distribution is highly inhomogeneous --- for instance, exhibiting 
heavy-tailed behavior. In such regimes, randomized strategies can yield order-of-magnitude 
reductions in gate count. However, this advantage is confined to moderate precision, 
typically around $\varepsilon \sim 10^{-3}$; at higher precision, deterministic methods 
consistently outperform their randomized counterparts.

The Hamiltonian ensembles considered here are designed to emphasize coefficient dispersion 
while suppressing additional structure, such as locality-induced commutation patterns. Since 
such structure is known to reduce Trotter error, the advantages of randomized methods 
reported in this work should be interpreted as an upper bound on their practical utility. 
In more structured physical Hamiltonians, deterministic approaches are expected to perform 
at least as well, and often better.

The limitations of randomized methods admit natural explanations in both algorithmic 
settings considered. For product-formula methods, achieving higher-order convergence without 
significant overhead remains difficult~\cite{qSWIFT}. For \ac{qsvt}-based algorithms, 
stochastic and approximation errors introduced at the block-encoding level are amplified 
through the polynomial transformation, producing an accuracy floor that limits the benefit 
of sparsification. More broadly, this highlights the sensitivity of \ac{qsvt}-based methods 
to imperfect oracle implementations, whether stochastic or variational in nature.

Several directions remain open for future work. A natural extension is to characterize how 
additional structural features --- commutation relations, symmetries, or inter-term 
correlations --- affect the performance of randomized methods. Extending this analysis to 
larger system sizes and more realistic Hamiltonians would further clarify the extent to 
which the trends observed here persist in practical settings.
\section*{Acknowledgment} 
FP and MG are supported by CERN through the CERN Quantum Technology Initiative. AG is supported by PNRR MUR projects under Grants PE0000023-NQSTI and CN00000013-ICSC, and by QUART\&T, a project funded by the Italian Institute of Nuclear Physics (INFN) within the Technological and Interdisciplinary Research Commission (CSN5) and the Theoretical Physics Commission (CSN4).

\bibliography{bibliography}

\appendix

\section{Error propagation in LCU}
\label{app:proof}

We analyze how imperfections in the \textsc{Prepare} and \textsc{Select} oracles propagate through the \ac{lcu} block-encoding and the subsequent \ac{qsvt} transformation. In particular, we show that errors in the LCU primitives lead to controlled (at most linear) deviations in the block-encoding, which are then amplified linearly in the QSVT polynomial degree.

\begin{theorem}[LCU error propagation]
\label{theo:LCUPropagationUnitary}
Given an approximate \textsc{Prepare} oracle $\tilde{V}$ such that
\begin{equation*}
\varepsilon_{\mathrm{prep}} = \|V - \tilde{V}\|,
\end{equation*}
the induced error on the LCU block-encoding
\begin{equation*}
\tilde{W} = (\tilde{V}^\dagger \otimes I)\, S\, (\tilde{V} \otimes I)
\end{equation*}
is bounded by
\begin{equation*}
\varepsilon_{\mathrm{LCU}} = \|W - \tilde{W}\| \le 2\,\varepsilon_{\mathrm{prep}},
\end{equation*}
where $W$ is the exact LCU block-encoding.
\end{theorem}

\begin{proof}
For notational simplicity, we omit identity operators:
\begin{align*}
\varepsilon_{\mathrm{LCU}} 
&= \|\tilde{V}^\dagger S \tilde{V} - V^\dagger S V\| \\
&= \|\tilde{V}^\dagger S (\tilde{V} - V) + (\tilde{V}^\dagger - V^\dagger) S V\| \\
&\le \|\tilde{V}^\dagger S\| \|\tilde{V} - V\| + \|S V\| \|\tilde{V}^\dagger - V^\dagger\| \\
&\le 2\,\varepsilon_{\mathrm{prep}},
\end{align*}
where we used the triangle inequality and the fact that $V$ and $S$ are unitary.
\end{proof}

An analogous argument shows that an error in the \textsc{Select} oracle $\varepsilon_{\mathrm{sel}} = \|S - \tilde{S}\|$ propagates linearly,
\begin{equation*}
\varepsilon_{\mathrm{LCU}} \le \varepsilon_{\mathrm{sel}}.
\end{equation*}

\begin{theorem}[QSVT error amplification]
\label{theo:QSVTPropagation}
Let $U_A$ be an exact block-encoding and $\tilde{U}_A$ an approximation satisfying
\begin{equation*}
\varepsilon_{\mathrm{BE}} = \|U_A - \tilde{U}_A\|.
\end{equation*}
Then a degree-$d$ \ac{qsvt} sequence satisfies
\begin{equation*}
\varepsilon_{\mathrm{QSVT}} \le d\,\varepsilon_{\mathrm{BE}}.
\end{equation*}
\end{theorem}

\begin{proof}
Write $\tilde{U}_A = U_A + \Delta U$ with $\|\Delta U\| = \varepsilon_{\mathrm{BE}}$. Expanding the QSVT sequence to first order in $\Delta U$, each application of the block-encoding contributes at most $\varepsilon_{\mathrm{BE}}$ to the total error. Since the sequence contains $d$ uses of the block-encoding and all intermediate operators are unitary, the total error is bounded by $d\,\varepsilon_{\mathrm{BE}}$.
\end{proof}

Finally, we characterize the dispersion of the stochastic Hamiltonian estimator used in the SparSto construction.

\begin{theorem}[Coefficient variance proxy]
\label{theo:EstimatorVariance}
Let
\begin{equation*}
\hat{H} = \sum_{j=1}^L \frac{c_j}{p_j} H_j \,\xi_j,
\end{equation*}
where $\xi_j \sim \mathrm{Bernoulli}(p_j)$ independently. Define
\begin{equation*}
\vec{u} = \left( \left(\frac{1}{p_1} - 1\right)c_1^2, \ldots, \left(\frac{1}{p_L} - 1\right)c_L^2 \right).
\end{equation*}
Then
\begin{equation*}
\mathrm{Var}_{\mathrm{coeff}}[\hat{H}] := \|\vec{u}\|_1
\end{equation*}
quantifies the dispersion of the coefficients entering the estimator.
\end{theorem}

\begin{proof}
Using independence of the Bernoulli variables,
\begin{align*}
\mathbb{E}[\hat{H}] &= H, \\
\mathbb{E}[\hat{H}^2] 
&= \sum_i \frac{c_i^2}{p_i} H_i^2 + \sum_{i \ne j} c_i c_j H_i H_j.
\end{align*}
Subtracting $H^2$ yields
\begin{align*}
\mathbb{E}[\hat{H}^2] - H^2 
&= \sum_i c_i^2\left(\frac{1}{p_i} - 1\right),
\end{align*}
which gives the stated expression.
\end{proof}

We emphasize that $\mathrm{Var}_{\mathrm{coeff}}[\hat{H}]$ is a scalar quantity derived from the coefficients and should not be interpreted as a full operator variance. Rather, it serves as a proxy controlling the magnitude of stochastic errors entering the block-encoding construction.

\section{Statistical characterization of molecular Hamiltonian}
\label{app:StatisticalHam}

In this appendix, we provide an empirical analysis of the coefficient statistics of molecular electronic-structure Hamiltonians and demonstrate that they are well captured by the heavy-tailed random coefficient models employed in the main text. The goal is to justify, quantitatively, the statistical assumptions underlying our random Hamiltonian construction.

\begin{figure}[tbp]
    \centering
    \includegraphics[]{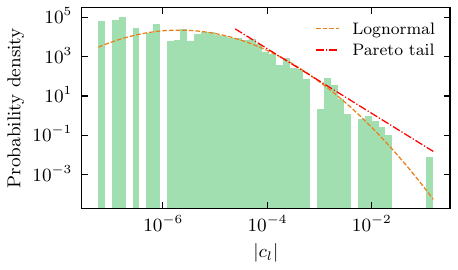}
    \caption{\justifying{\textbf{Coefficient distribution} of the $\mathrm{HCN}$ molecular Hamiltonian. The histogram represents the normalized coefficient magnitudes $|c_l|$, together with a lognormal fit of the bulk of the distribution and a Pareto fit of the tail.}}
    \label{fig:FitHCN}
\end{figure}

\begin{table}[thb]
    \centering
    \begin{tabular}{ccccc}
    \hline
    Molecule  & $L$ & $N_\text{qubits}$ & Lognormal $\sigma^2$ & Pareto $a$\\
    \hline
    $\mathrm{H_2}$ &  15 & 4 & 0.60& 2.14\\
    $\mathrm{H_4}$ &  185 & 8 & 1.20& 0.96\\
    $\mathrm{LiH}$ &  631 & 12 & 3.09& 0.75\\
    $\mathrm{BeH_2}$ &  666 & 14 & 2.62& 0.68\\
    $\mathrm{CH_2}$ &  1086 & 14 & 2.85& 0.68\\
    $\mathrm{H_2O}$ &  1086 & 14 & 2.75& 0.62\\
    $\mathrm{NH_3}$ &  2281 & 16 & 3.92& 0.76\\
    $\mathrm{CH_4}$ &  6892 & 18 & 4.24& 0.63\\
    $\mathrm{C_2}$ &  2951 & 20 & 3.34& 0.61\\
    $\mathrm{O_2}$ &  2239 & 20 & 2.72& 0.63\\
    $\mathrm{N_2}$ &  2951 & 20 & 3.02& 0.63\\
    $\mathrm{HCN}$ &  8758 & 22 & 4.28& 0.54\\
    $\mathrm{H_2O_2}$ &  8865 & 24 & 3.20& 0.59\\
    $\mathrm{C_2H_4}$ &  8919 & 28 & 2.59& 0.64\\
    $\mathrm{CH_2O}$ &  17657 & 24 & 4.80& 0.53\\
    $\mathrm{N_2H_4}$ &  27735 & 28 & 3.02& 0.60\\
    $\mathrm{CO_2}$ &  16170 & 30 & 5.71& 0.49\\
    $\mathrm{C_2H_6}$ &  22337 & 32 & 17.89& 0.55\\
    \hline
    \end{tabular}
 \caption{\justifying{\textbf{Molecular Hamiltonian characteristics.} 
Number of Pauli terms $L$, number of qubits $N_{\mathrm{qubits}}$ after Jordan–Wigner mapping, and fitted parameters of the coefficient-magnitude distributions. For the lognormal model, $\sigma^2=\mathrm{Var}(\log |c_l|)$ quantifies dispersion; for the Pareto-II (Lomax) model, $a$ is the tail shape parameter. All molecules are given in the STO-3G basis.}}
    \label{tab:MoleculeSigma}
\end{table}

For each molecule, the electronic Hamiltonian from the PennyLane datasets~\cite{PennylaneChemistryDatasets} is mapped to qubits via the Jordan--Wigner transformation and rescaled such that $\sum_l |c_l| = 1$, isolating structural features of the coefficient distribution from the overall energy scale. We model the empirical distribution of coefficient magnitudes ${|c_l|}$ using either a lognormal or a Pareto-II (Lomax) distribution. In the lognormal case,
\begin{equation}
\log |c_l| \sim \mathcal{N}(\mu,\sigma^2),
\end{equation}
where $\sigma^2=\mathrm{Var}[\log |c_l|]$ quantifies dispersion. In the Pareto-II case,
\begin{equation}
\mathrm{Pareto\text{-}II}(x;a)=a(1+x)^{-(a+1)},
\end{equation}
with shape parameter $a$ controlling tail heaviness, where smaller $a$ corresponds to stronger dominance by large coefficients. The fitted parameters $\sigma^2$ and $a$ provide quantitative measures of coefficient inhomogeneity, which is central for randomized simulation methods that exploit uneven weight distributions. The values reported in Table~\ref{tab:MoleculeSigma} show pronounced heterogeneity across molecules, with lognormal variances and Pareto shape parameters lying in the same range as those used in our numerical experiments. This statistical consistency indicates that the degree of coefficient dispersion and tail heaviness in our synthetic ensembles closely matches that observed in realistic quantum chemistry Hamiltonians. Consequently, the heavy-tailed random Hamiltonian model in ~\eqref{eq:RandomHamiltonianExtraction} serves as a meaningful proxy for chemically motivated systems when assessing the performance of deterministic and randomized simulation methods.

\section{Gate counting}
\label{app:gate_counts}

The most relevant performance metric is the resource requirement—particularly gate counts—rather than the parameters $r$ or $d$ alone. Many optimization techniques, such as sparsification, primarily affect the proportionality constant between gate count and $r$ or $d$, rather than reducing these parameters directly. In the following, we describe how we count the gates in each cases. 

We assume an all-to-all connected, fault-tolerant quantum device implementing gates from the Clifford+T universal gate set
\[
\text{Clifford+T} = \{\text{S}, \text{H}, \text{CNOT}, \text{T}\},
\]
and focus on the number of CNOT and T gates. CNOT gates are the only multi-qubit gates in the set and are essential for generating entanglement and implementing multi-qubit rotations; they are also among the most expensive gates experimentally~\cite{CNOT_expenses}. T gates serve as the standard metric for fault-tolerant cost, as they require magic-state distillation~\cite{ChildsToward,MagicStateDistillation,SurfaceCodes}. The remaining Clifford gates are not explicitly counted, as they can be efficiently simulated classically~\cite{SimulatingCliffordGates}.

\subsubsection{Rotation decomposition gate count}

Since the Clifford+T set contains only discrete gates, continuous rotations such as $R_z(\theta)$ must be approximated to accuracy $\varepsilon_{\mathrm{deco}}$ using sequences of T, S, and Hadamard gates. We employ the gridsynth (Ross--Selinger) algorithm~\cite{RossSelingerAlgorithm}. Except for angles that are integer multiples of $\pi/4$, the T gate cost is essentially independent of $\theta$ at fixed $\varepsilon_{\mathrm{deco}}$.

To ensure a fair comparison between different algorithms, we choose the decomposition accuracy as a function of the total number of $R_z$ rotations,
\[
\varepsilon_{\mathrm{deco}} = \frac{\varepsilon}{\text{No. of } R_z \text{ rotations}}.
\]
All required rotations are compiled using the implementation of~\cite{RossSelingerAlgorithm}, thereby accounting for the induced T gate overhead.

\subsubsection{Trotter gate count}

From ~\eqref{eq:Trotter1}, the circuits produced by Trotterization, SparSto, and qDRIFT consist of factors of the form $e^{-i\theta H_\ell}$, where $H_\ell$ is a Pauli string of length $K$. Such operators can be compiled using Clifford gates and a single $R_z$ rotation~\cite{Tacchino2019}. Specifically, each $\sigma_x$ and $\sigma_y$ operator in $H_\ell$ is converted to $\sigma_z$ via basis changes using $H$ and $HS$, respectively, yielding a rotation of the form $e^{-i\theta Z^{\otimes K}}$. A CNOT ladder is then used to reduce $Z^{\otimes K}$ to a single $Z$, resulting in an $R_z(\theta)$ rotation. The ladder and basis changes are subsequently uncomputed, and $R_z(\theta)$ is synthesized using gridsynth.

This procedure yields the exact cost of a single Trotter step; multiplying by $r$ gives the total Trotterization cost. For SparSto and qDRIFT, the terms applied at each step are sampled stochastically. However, since the decomposition cost of $R_z(\theta)$ is approximately independent of $\theta$, the average cost per step is obtained by weighting the cost of each term by its sampling probability $p_\ell$.

\subsubsection{\ac{qsvt} gate count}

The compilation of \ac{qsvt} and Sparse \ac{qsvt} circuits in the Clifford+T gate set naturally separates into two contributions: 
(i) the projector-controlled phase operators $\Pi_\phi$ and $\tilde{\Pi}_\phi$, and 
(ii) the \ac{lcu} block-encoding of the Hamiltonian.

For (i), each projector-controlled phase appearing in ~\eqref{eq:ProjectorControlledPhases} can be decomposed as~\cite{GilBeyond}
\begin{equation}
    \begin{split}
        &\mathrm{C_\Pi NOT}\;\otimes\; R_z(\phi)\;\otimes\;\mathrm{C_\Pi NOT}, \text{ with }\\
        &\mathrm{C_\Pi NOT} = X\otimes \Pi + \mathbb{I}\otimes(\mathbb{I}-\Pi).
    \end{split}
\end{equation}
In our construction, the Hamiltonian is block-encoded in the subspace {Select}ed by
\[
\Pi=\tilde{\Pi}=(\ket{0}\bra{0})^{\otimes a},
\]
so that $\mathrm{C_\Pi NOT}$ corresponds to a multi-controlled $X$ gate acting on $a$ ancilla qubits. We decompose this operation into Toffoli gates, each implemented using $6$ CNOTs, $2$ Hadamard gates, and $7$ T gates~\cite{ToffoliImplementation}, together with additional CNOT chains required to implement multiple controls~\cite{MultiControlledX1, MultiControlledX2}. The single-qubit rotation $R_z(\phi)$ is synthesized using the gridsynth (Ross--Selinger) algorithm. As a result, for fixed ancilla register size $a$ and rotation synthesis accuracy $\varepsilon_\text{deco}$, each projector-controlled phase contributes a fixed and explicitly determined number of T and CNOT gates.

For (ii), the \ac{lcu} block-encoding consists of two applications of the state-preparation oracle $V$ and one application of the \textsc{Select} oracle $S$. The oracle $V$ is implemented using the M\"ott\"onen state-preparation algorithm~\cite{MottonenStatePreparation}, whose CNOT and single-qubit rotation counts are known explicitly as functions of the number of ancilla qubits $a=\lceil\log_2 L\rceil$. The \textsc{Select} oracle $S$ implements a structured multi-controlled operation, whose optimized gate complexity is upper-bounded by the binary-tree walk over Boolean products of the control qubits~\cite[Section~G.4, Supporting Information]{ChildsToward}. 

Therefore, for fixed $a$ and QSVT polynomial degree $d$, the total Clifford+T cost of \ac{qsvt} (and Sparse \ac{qsvt}) is obtained by summing the contributions from all $d$ projector-controlled phase operators together with the required instances of the $V$ and $S$ oracles. To account for the sparsified gate count, the number of ancilla qubits is computed from the average number of terms per step $\mu$ as $a = \lceil\log_2\mu\rceil$.

\end{document}